\documentclass{ro}
\usepackage{graphicx}
\usepackage{multirow}
\usepackage{amsmath,amssymb,amsfonts}
\usepackage{amsthm}
\usepackage{mathrsfs}
\usepackage[title]{appendix}
\usepackage{xcolor}
\usepackage{textcomp}
\usepackage{manyfoot}
\usepackage{booktabs}
\usepackage{algorithm}
\usepackage{algorithmicx}
\usepackage{algpseudocode}
\usepackage{listings}
\usepackage{float}
\usepackage{graphicx}
\usepackage{tabularx}
\usepackage{multirow}
\usepackage{rotating}
\usepackage{array}
\usepackage{longtable}
\usepackage{url}
\usepackage{tikz}
\usepackage{pgfplots}
\usepackage{pgfplotstable}
\usepackage{pdflscape}
\usepackage{adjustbox}
\usepgfplotslibrary{groupplots}
\pgfplotsset{compat=1.18}

\makeatletter
\def\@setdate{}
\makeatother

\newtheorem{theorem}{Theorem}

\begin{document}
\title{Proving Optimality for the Bandwidth Multicoloring Problem via SAT}
%
\author{Duc Trung Kim Nguyen}
\address{VNU University of Engineering and Technology, Hanoi, Vietnam }
\email{23021533@vnu.edu.vn}
\author{Khanh Van To}\sameaddress{1} \email{khanhtv@vnu.edu.vn}
\date{}
\begin{abstract}
The Bandwidth Multicoloring Problem (BMCP) is an NP-hard extension of the Bandwidth Coloring Problem (BCP) with important applications in telecommunications, resource allocation, and scheduling. While state-of-the-art metaheuristics can efficiently produce high-quality solutions, they cannot certify global optimality. Existing exact approaches based on Constraint Programming (CP) and Integer Programming (IP) provide such guarantees but typically require extensive computation and still lag behind metaheuristics in solution quality, leaving many benchmark instances without optimality certificates.
In this paper, we present the first SAT-based exact framework for the BMCP. Our main contribution is an efficient SAT encoding that compactly models both intra-vertex and inter-vertex color distance constraints. Combined with tight color domain reduction and an incremental SAT-solving strategy, the proposed formulation significantly prunes the search space and enables efficient exact optimization.
Experimental results on the GEOM and MS-CAP benchmark suites demonstrate substantial improvements over previous exact approaches. On the challenging GEOM benchmark, the proposed framework proves optimality for more instances within only one hour of computation than the previous CP/IP approach, which required a 48-hour time limit, while also verifying the optimality of several previously reported best-known solutions.
These results demonstrate that SAT-based reasoning provides an effective exact optimization framework for the BMCP and substantially expands the range of benchmark instances whose optimality can be certified.
\end{abstract}

\subjclass{05C15, 90C27, 68R07}

\keywords{Channel assignment, Graph coloring problem, Bandwidth coloring, Bandwidth Multicoloring, SAT solving}

%
%
\maketitle
\section{Introduction}
\label{sec:introduction}

The Bandwidth Coloring Problem (BCP) is a combinatorial optimization problem that extends the classical graph coloring problem (GCP), with numerous practical applications in areas such as scheduling, resource allocation, and network design. The problem can be formally defined as follows: given an edge-weighted undirected graph $G = (V,E)$, a weight function $d: E \to \mathbb{Z}^+$ assigns a positive integer weight to each edge $e = \{u,v\} \in E$. The objective of the BCP is to find a coloring $c: V \to \mathbb{N}$ such that $\forall e = \{u,v\} \in E: |c(u) - c(v)| \geq d(\{u,v\})$, while minimizing the span value, defined as $k = \max_{v \in V} c(v)$. 

The Bandwidth Multicoloring Problem (BMCP) generalizes the BCP by allowing each vertex to be assigned multiple colors. Formally, let $G = (V,E)$ be an undirected graph equipped with two weight functions: a vertex weight function $w: V \to \mathbb{Z}^+$ specifying the number of colors required for each vertex, and an edge weight function $d: E \to \mathbb{Z}^+$ defining the minimum color separation between adjacent vertices. The objective of the BMCP is to find a multicoloring assignment $c: V \rightarrow \mathcal{P}(\mathbb{Z}^+)$ that minimizes the span value $k = \max_{u \in V, m\in c(u)} m $ while satisfying the following conditions: first, each vertex $u \in V$ is assigned exactly $w(u)$ distinct colors; second, any two distinct colors $m, n \in c(u)$ assigned to the same vertex must satisfy $|m - n| \geq d(\{u,u\})$; and finally, for any edge $e = \{u,v\} \in E$, every color $m \in c(u)$ and $n \in c(v)$ must satisfy $|m - n| \geq d(\{u,v\})$.

Both the BCP and the BMCP arise naturally from frequency assignment problems (FAP) \cite{aardal2007} in the field of telecommunications, where the graph's vertices represent radio transmitters and the edge weights dictate the minimum frequency separations required to avoid electromagnetic interference \cite{eisenblatter2002}. Because both models are generalizations of the classical graph coloring problem, they are strongly NP-hard \cite{malaguti2010survey}. Consequently, a vast amount of literature has been dedicated to metaheuristic approaches to find high-quality approximations, including iterated tabu search \cite{lai2013multistart}, learning-based hybrid search \cite{jin2015}, path relinking \cite{lai2016}, and variable neighborhood search \cite{matic2017}. While these methods perform well on large-scale instances, optimality matters profoundly in combinatorial optimization. Establishing mathematically proven optimal solutions is not merely a theoretical exercise; it is a necessity for benchmark certification and for reliably validating heuristic results. Furthermore, providing exact lower bounds is crucial for supporting future algorithm development. The critical importance of exact methods is perhaps best highlighted by their ability to identify incorrect published solutions, as demonstrated when exact approaches successfully proved that a previously published heuristic solution for the BMCP \cite{chakraborty2001efficient} was mathematically flawed \cite{dias2016constraint}. 

Despite this pressing need for provable optimality, existing exact methods for the BMCP have essentially stagnated since 2021. The current state-of-the-art exact solvers rely heavily on Constraint Programming (CP) and Integer Programming (IP) models introduced by Dias et al. \cite{dias2016constraint, dias2021}. While these models can handle small-scale instances, they suffer from exponential computational complexity and fail to scale effectively. Conversely, the adjacent field of the BCP has recently experienced a major breakthrough. Boolean Satisfiability (SAT)-based methods have dramatically changed the state-of-the-art for the BCP \cite{dey2013, heule2022}. Most notably, Faber et al. \cite{faber2024} introduced partial-ordering-based SAT encodings (POP-S-B and POPH-S-B) that successfully identified optimal solutions for 32 out of 33 instances in the challenging GEOM dataset \cite{trick2002}. Surprisingly, despite the massive success of SAT solving in the BCP, the BMCP has never benefited from SAT methodologies. Extending these formulations to the BMCP is highly non-trivial. Naive adaptations, such as transforming the BMCP into a BCP via vertex splitting, inevitably lead to an unmanageable explosion in the number of variables and clauses, severely crippling solver efficiency. 

This paper fills that methodological gap by proposing the first dedicated SAT-based framework designed natively for the complex combinatorial structure of the BMCP. The primary contributions of this work are summarized as follows. First, we introduce the very first SAT-based framework for the BMCP, completely bypassing the inefficiencies of vertex decomposition. As our primary contribution, we propose a novel Order-Based Encoding (OBE) that natively captures the problem's complex combinatorial structure. Second, our proposed methodology significantly improves upon the state-of-the-art in exact solving, decisively outperforming existing CP and IP formulations. Third, by leveraging this encoding, we successfully prove optimality for several previously open instances within the benchmark datasets. Finally, and most importantly, our SAT formulation achieves these unprecedented exact solutions under a much smaller computational budget than what was previously required by the existing exact models, underscoring the exceptional efficiency and scalability of our approach.

The remainder of this paper is organized as follows. Section \ref{sec:related-work} reviews related work on metaheuristic and exact approaches for the BMCP. Section \ref{sec:sat-formula} presents the proposed SAT-based approach, describing the encoding method. Section \ref{sec:experiments} reports computational experiments on the GEOM and MS-CAP benchmark instances. Finally, Section \ref{sec:conclusion} summarizes the findings and discusses directions for future research.

\section{Related work}
\label{sec:related-work}
The Bandwidth Coloring Problem and its generalization, the Bandwidth Multicoloring Problem, model resource allocation scenarios where interacting elements must be separated by predefined distances. While the BCP assigns a single color to each vertex such that adjacent vertices respect specific edge-weight separations, the BMCP introduces the additional complexity of demanding multiple colors per vertex, governed by internal vertex separation constraints. Because these problems are NP-hard, the pursuit of optimal or near-optimal solutions has motivated a rich body of literature. Existing methodologies can be broadly categorized into heuristic or metaheuristic techniques, which prioritize computational speed for large instances, and exact methods, which are mathematically designed to guarantee solution optimality. Table \ref{tab:previous-work} provides an overview of existing studies, including their key contributions, the problems addressed, and the availability of optimality guarantees for the BCP and the BMCP. 

\begin{table}[htpb]
\centering
\caption{Previous work on BCP and BMCP}
\label{tab:previous-work}
\begin{adjustbox}{width=\linewidth}
\begin{tabular}{|c|p{5cm}|l|l|} 
\hline
\textbf{Work} & \multicolumn{1}{c|}{\textbf{Main contribution}} & \multicolumn{1}{c|}{\textbf{Target Problem}} & \multicolumn{1}{c|}{\begin{tabular}[c]{@{}c@{}}\textbf{Proof of}\\\textbf{Optimality}\end{tabular}}  \\ 
\hline
\begin{tabular}[c]{@{}c@{}}Faber et al.,\\2024 \cite{faber2024}\end{tabular} & Partial-ordering SAT models (POP-S-B and POPH-S-B) & BCP & Yes \\
\hline
\begin{tabular}[c]{@{}c@{}}Lai et al.,~\\2013 \cite{lai2013multistart}\end{tabular} & Multistart Iterated Tabu Search (MITS) algorithm & BCP, BMCP & No \\ 
\hline
\begin{tabular}[c]{@{}c@{}}Jin et al.,\\2014 \cite{jin2014effective}\end{tabular} & Learning-based Hybrid Search (LHS) algorithm & BCP, BMCP & No \\ 
\hline
\begin{tabular}[c]{@{}c@{}}Lai et al.,~\\2016 \cite{lai2016}\end{tabular} & Path relinking strategies & BCP, BMCP & No \\ 
\hline
\begin{tabular}[c]{@{}c@{}}Mati\'c et al.,\\2017 \cite{matic2017}\end{tabular} & Variable Neighborhood Search (VNS) & BCP, BMCP & No \\ 
\hline
\begin{tabular}[c]{@{}c@{}}Dias et al.,\\2016, 2021 \cite{dias2016constraint,dias2021}\end{tabular} & Constraint Programming (CP) and Integer Programming (IP) models & BCP, BMCP & Yes \\ 
\hline
\end{tabular}
\end{adjustbox}
\end{table}

Metaheuristic algorithms have been extensively applied to both the BCP and the BMCP, primarily due to their ability to efficiently generate high-quality solutions. To tackle the BMCP using these methods, the problem is typically reduced to a standard BCP by decomposing each vertex into a clique, where the clique size corresponds to the vertex's color demand, and internal edge weights equal the vertex's minimum separation distance. Several prominent frameworks have been successfully adapted for these problems. These include the Multistart Iterated Tabu Search (MITS) \cite{lai2013multistart}, which leverages fast incremental evaluation to streamline objective function updates; Learning-based Hybrid Search (LHS) \cite{jin2014effective}, which combines forward-checking construction with adaptive tabu search repair; path relinking strategies \cite{lai2016}, which maintain a balance between search intensification and diversification; and Variable Neighborhood Search (VNS) \cite{matic2017}, which structures local search around conflict metrics and edge weights.

Despite their computational speed and scalability, these metaheuristic approaches possess inherent limitations, most notably their fundamental inability to guarantee global optimality. As previously established, proving optimality is not merely a theoretical exercise but a critical necessity for advancing the field. Exact mathematical guarantees are indispensable for rigorously certifying benchmarks, reliably validating the outputs of heuristic algorithms, and providing exact lower bounds that support future algorithmic development. Moreover, such rigor is uniquely capable of identifying and correcting erroneously published solutions. In real-world applications such as telecommunications frequency assignment, the absence of these exact guarantees can have severe practical consequences; relying on heuristic algorithms, which are prone to stagnating in local optima and heavily dependent on extensive parameter tuning, often yields suboptimal allocations, ultimately leading to costly spectrum waste and signal interference.

To address the need for mathematical guarantees, exact formulations for the BMCP have been explored through the lens of Constraint Programming and Integer Programming. Dias et al. \cite{dias2016constraint,dias2021} proposed both CP and IP models designed to find provably optimal solutions. In their CP model, they utilize integer variables $x(i,k)$ to represent the $k$-th color assigned to vertex $i$ ($1 \leq k \leq w(i)$). This is supported by a global constraint, \textit{allMinDistance}, which ensures that the minimum distance between any two colors assigned to the same vertex is at least the vertex's minimum separation distance. The CP formulation is given by:
\begin{align}
\text{Minimize} \quad
    & \max_{\substack{i \in V \\ 1 \le k \le w(i)}} x(i,k)
\\
\text{Subject to} \quad
    & |x(i,k)-x(j,m)| \ge d(\{i,j\})
    && \forall \{i,j\}\in E,
\\
    & && 1 \le k \le w(i) \notag
\\
    & && 1 \le m \le w(j) \notag
\\
    & \text{allMinDistance}\!\left(
        \{x(i,k):1\le k\le w(i)\},
        d(\{i,i\})
      \right)
    && \forall i\in V
\\
    & x(i,k)\in\mathbb{Z}^{+}
    && \forall i\in V,
\\
    & && 1 \le k \le w(i) \notag
\end{align}

Alternatively, their IP model introduces two sets of variables: $x_{ic}$, a binary variable indicating whether color $c$ is assigned to vertex $i$, and $z_{\max}$, an integer variable representing the maximum color used. The IP formulation is expressed as:
\begin{align}
\text{Minimize} \quad
    & z_{\max}
\\[0.5ex]
\text{Subject to} \quad
    & \sum_{c=1}^{UB} x_{ic} = w(i)
    &&
    \forall i \in V
\\
    & x_{ic} + x_{je} \le 1
    &&
    \begin{aligned}[t]
        &\forall \{i,j\} \in E, i \neq j,\\
        &1 \le c,e \le UB,\\
        &|c-e| < d(\{i,j\})
    \end{aligned}
\\
    & x_{ic} + x_{ie} \le 1
    &&
    \begin{aligned}[t]
        &\forall i \in V,\\
        &1 \le c,e \le UB,\\
        &|c-e| < d(\{i,i\})
    \end{aligned}
\\
    & z_{\max} \ge c\,x_{ic}
    &&
    \begin{aligned}[t]
        &\forall i \in V,\\
        &1 \le c \le UB
    \end{aligned}
\\
    & x_{ic} \in \{0,1\}
    &&
    \begin{aligned}[t]
        &\forall i \in V,\\
        &1 \le c \le UB
    \end{aligned}
\\
    & z_{\max} \in \mathbb{R}
\end{align}

Their comparative studies of these formulations show that the IP model generally outperforms the CP model in terms of computational efficiency, particularly for larger instances. While the IP model is more computationally efficient at proving optimality, they note that the CP model still holds a specific advantage in certain scenarios. If the solvers hit the designated time limit without finding the optimal solution, the CP model tends to return a higher-quality feasible solution than the IP model. However, both models face significant challenges when scaling to larger problem sizes due to the exponential growth of the search space. 

By translating complex combinatorial constraints into a unified Boolean formula, SAT methods leverage the highly optimized inference engines of modern SAT solvers. Because these solvers utilize advanced techniques such as Conflict-Driven Clause Learning (CDCL) \cite{marques2009conflict}, which allows the algorithm to learn from dead ends and drastically prune the search tree, they excel at navigating tightly constrained problem spaces. For the BCP, the most significant recent advancement was made by Faber et al. (2024) \cite{faber2024}, who proposed highly effective partial-ordering-based SAT encodings (POP-S-B and POPH-S-B). By generalizing earlier ordering approaches, their encodings optimally solved several previously open GEOM benchmark instances for the first time. 

Consequently, the current state of the literature reveals a significant research gap in solving the BMCP. First, while metaheuristic algorithms produce excellent feasible solutions, they inherently lack the mathematical mechanisms required to prove optimality. Second, exact approaches such as the CP and IP models developed by Dias et al. \cite{dias2016constraint,dias2021} successfully provide these strict optimality guarantees, yet their scalability remains highly restricted, severely limiting their application to larger problem instances. Finally, despite the fact that SAT-based methods demonstrated tremendous success and computational efficiency for the related BCP in 2024, there is currently no existing SAT encoding formulated for the BMCP. Bridging this gap requires more than trivial adaptations; relying on naive reductions, such as transforming a BMCP instance into a standard BCP by splitting vertices into dense cliques, inherently triggers a catastrophic proliferation of variables and clauses that neutralizes the advanced CDCL mechanisms. Therefore, there is an imperative need to formulate compact, native SAT encodings that directly capture both intra-vertex and inter-vertex distance constraints to effectively harness the deductive power of SAT solvers for the BMCP.

\section{SAT Encodings for the Bandwidth Multicoloring Problem} \label{sec:sat-formula}
In this section, we develop three SAT formulations for the BMCP. To establish our framework, we first introduce a Flat Encoding (FE) as a direct baseline, followed by an intermediate Slot-Based Encoding (SBE) that utilizes local bounds. Finally, we present the core contribution of this paper: the highly optimized Order-Based Encoding (OBE), which elegantly overcomes the spatial and computational limitations of the prior models.

\subsection{Global and Local Bound Computation} \label{subsec:bounds}

To reduce the search space and accelerate SAT-based solving, we establish tight lower and upper bounds for the span, as well as restricted color domains for individual vertices.

The minimum span is constrained by the densest substructures within the graph. For any clique $C \subseteq V$, the required span is at least $1 + (W_C - 1) \cdot d_{min}^C$, where $W_C = \sum_{v \in C} w(v)$ is the total color requirement and $d_{min}^C = \min_{u, v \in C} d(\{u,v\})$ is the tightest distance constraint in $C$. We define the lower bound $LB$ as the maximum bound discovered across maximal cliques, which are extracted using the Bron-Kerbosch algorithm \cite{bron1973algorithm}. A formal proof is provided in Appendix \ref{apx:lower-bound-correctness}.

Conversely, an upper bound ($UB$) caps the SAT variables' domains. The trivial bound $1 + \left( \sum_{u \in V} w(u) - 1 \right) \cdot \max_{\{u,v\} \in E} d(\{u,v\})$ is overly large and yields intractable SAT formulas. Instead, we compute a tighter $UB$ via a greedy heuristic \cite{dias2021} (Algorithm \ref{alg:bmcp_upper_bound}). By prioritizing vertices with the largest weight $w(u)$, the algorithm sequentially assigns colors, incrementally increasing candidate values until all distance constraints are satisfied. The maximum color utilized by this heuristic becomes our global $UB$.

Given a target span $k \in [LB, UB]$, we restrict the allowed color domains for each vertex to preemptively prune invalid assignments. To explicitly break symmetries, we enforce a strict ascending order for the colors assigned to any vertex $u$: $c(u)_1 < c(u)_2 < \dots < c(u)_{w(u)}$. This ordering dictates tight bounds on the minimum ($LB_{u,i}$) and maximum ($UB_{u,i}$) possible values for the $i$-th color of vertex $u$:
\begin{alignat*}{2}
LB_{u,i}
    &= 1 + (i-1)\cdot d(\{u,u\})
    &\qquad&
    \forall u \in V,\; i \in [1,w(u)]
\\
UB_{u,i}
    &= k - (w(u)-i) \cdot d(\{u,u\})
    &\qquad&
    \forall u \in V,\; i \in [1,w(u)]
\end{alignat*}

If $UB_{u,i} < LB_{u,i}$ for any vertex $u$ and index $i$, the targeted span $k$ is mathematically too small to accommodate $u$'s intra-vertex constraints. Consequently, we can immediately deduce that $k$ is invalid without generating or evaluating the SAT formula.

\begin{algorithm}[htb]
\caption{Greedy Upper Bound Computation for BMCP}
\label{alg:bmcp_upper_bound}
\begin{algorithmic}[1]
\Require Graph $G = (V, E)$ with vertex weights $w$ and edge weights $d$
\Ensure Upper bound $UB$ on the optimal span
\State $V' \gets V$
\State $numCol(i) \gets 0$ for all $i \in V$
\State $c(i) \gets \emptyset$ for all $i \in V$
\While{$V' \neq \emptyset$}
    \State $u \gets \arg\max_{v \in V'} w(v)$ \Comment{Select unassigned vertex with max weight}
    \While{$numCol(u) < w(u)$}
        \State $candColor \gets numCol(u) \times d(\{u,u\}) + 1$
        \State $violated \gets \text{true}$
        \While{$violated$}
            \State $violated \gets \text{false}$
            \For{$v \in (V \setminus V') \cup \{u\}$}
                \For{$k \in c(v)$}
                    \If{$|k - candColor| < d(\{u,v\})$}
                        \State $violated \gets \text{true}$
                        \State $candColor \gets candColor + 1$
                        \State \textbf{break}
                    \EndIf
                \EndFor
                \If{$violated$}
                    \State \textbf{break}
                \EndIf
            \EndFor
        \EndWhile
        \State $c(u) \gets c(u) \cup \{candColor\}$
        \State $numCol(u) \gets numCol(u) + 1$
    \EndWhile
    \State $V' \gets V' \setminus \{u\}$
\EndWhile
\State $UB \gets \max_{v \in V, k \in c(v)} k$
\State \Return $UB$
\end{algorithmic}
\end{algorithm}

\subsection{Flat Encoding (FE)} \label{subsec:flat-encoding}
Serving as our initial baseline, this flat SAT encoding is a direct transformation of the IP model proposed by Dias et al. \cite{dias2016constraint,dias2021}. The encoding uses Boolean variables for each vertex $u \in V$ and each available color $m \in \{1, 2, \dots, k\}$ with the semantics that $x_{u,m} = \text{true} \iff m\in c(u)$.

The constraints are defined as follows:
\begin{itemize}
    \item \textit{Color cardinality constraint:}
        \begin{equation}
            \sum_{m=1}^k x_{u, m} = w(u) \quad \forall u \in V
        \end{equation}
    
    \item \textit{Intra-vertex color distance constraint:}
        \begin{equation}
            \neg x_{u, m} \lor \neg x_{u, n} \quad 
            \begin{aligned}
                &\forall u \in V, \; \forall m,n \in \{1, \dots, k\} \\
                &\text{s.t. } m < n \text{ and } n - m < d(\{u,u\}) 
            \end{aligned}
        \end{equation}
    
    \item \textit{Inter-vertex color distance constraint:}
        \begin{equation}
            \neg x_{u,m} \lor \neg x_{v,n} \quad 
            \begin{aligned}
                &\forall \{u,v\} \in E, u \neq v, \; \forall m,n \in \{1, \dots, k\} \\
                &\text{s.t. } |m - n| < d(\{u,v\})
            \end{aligned}
        \end{equation}
\end{itemize}

\subsection{Slot-Based Encoding (SBE)} \label{subsec:slot-based-encoding}
The flat encoding does not take advantage of the restricted color domains computed for each vertex, as it evaluates all colors up to the target span $k$. To address this issue and build towards our primary formulation, this encoding leverages the local bounds to exploit the strict ascending order for the colors assigned to any vertex. It employs Boolean variables $x_{u,i,m} = \text{true} \iff c(u)_i = m \quad \forall u \in V, i \in [1, w(u)], m \in [LB_{u,i}, UB_{u,i}]$.

The constraints for this encoding are:
\begin{itemize}
    \item \textit{Exactly-one value for the $i$-th color:}
        \begin{equation}
            \sum_{m=LB_{u,i}}^{UB_{u,i}} x_{u, i, m} = 1 \quad \forall u \in V, \forall i \in [1, w(u)]
        \end{equation}
    
    \item \textit{Intra-vertex color distance constraint:}
            \begin{equation}
                \neg x_{u,i,m} \lor \neg x_{u,i+1,n} \quad
                \begin{aligned}
                    &\forall u \in V, \; \forall i \in [1, w(u)-1], \\
                    &\forall m \in [LB_{u,i}, UB_{u,i}], \; \forall n \in [LB_{u,i+1}, UB_{u,i+1}] \\
                    &\text{s.t. } n - m < d(\{u,u\}) 
                \end{aligned}
            \end{equation}
        
        \item \textit{Inter-vertex color distance constraint:}
            \begin{equation}
                \neg x_{u,i,m} \lor \neg x_{v,j,n} \quad
                \begin{aligned}
                    &\forall \{u,v\} \in E, u \neq v, \; \forall i \in [1, w(u)], \; \forall j \in [1, w(v)] \\
                    &\forall m \in [LB_{u,i}, UB_{u,i}], \; \forall n \in [LB_{v,j}, UB_{v,j}] \\
                    &\text{s.t. } |m - n| < d(\{u,v\})
                \end{aligned}
            \end{equation}
\end{itemize}

The \textit{Color cardinality constraint} is inherently satisfied by instantiating variables over the explicit index $i \in [1, w(u)]$, which naturally guarantees that each vertex $u$ is assigned exactly $w(u)$ colors, making an explicit encoding of this constraint unnecessary.

\subsection{Order-Based Encoding (OBE)} \label{subsec:order-based-encoding}
Representing the primary contribution of this work, we extend the partial-ordering concepts proposed by Faber et al. \cite{faber2024} and our intermediate slot-based approach to develop the Order-Based Encoding (OBE) for the BMCP. The encoding uses the following variables:
\begin{itemize}
    \item $x_{u,i,m} = \text{true} \iff c(u)_i=m$
    \item $g_{u,i,m} = \text{true} \iff c(u)_i\geq m$
\end{itemize}
The constraints are defined as follows:
\begin{itemize}
    \item \textit{Linkage between variables $x$ and $g$:}
            \begin{alignat}{2}
                x_{u,i,m}
                &\leftrightarrow g_{u,i,m} \wedge \neg g_{u,i,m+1}
                &\qquad&
                \begin{aligned}[t]
                    &\forall u \in V,\; i \in [1,w(u)],\\
                    &m \in [LB_{u,i},UB_{u,i}-1]
                \end{aligned}
                \\
                x_{u,i,UB_{u,i}}
                &\leftrightarrow g_{u,i,UB_{u,i}}
                &\qquad&
                \forall u \in V,\; i \in [1,w(u)]
            \end{alignat}

    \item \textit{Monotonicity of variable $g$:}
            \begin{alignat}{2} \label{eq:monotonicity}
                g_{u,i,m} \rightarrow g_{u,i,m-1}
                &\qquad&
                \begin{aligned}[t]
                    &\forall u \in V,\; i \in [1,w(u)],\\
                    &m \in [LB_{u,i}+1,UB_{u,i}]
                \end{aligned}
            \end{alignat}

    \item \textit{Lower bound of the color value assigned to each vertex:}
            \begin{alignat}{2}
                g_{u,i,LB_{u,i}} \qquad \forall u \in V, i \in [1, w(u)]
            \end{alignat}

    \item \textit{Intra-vertex color distance constraint:}
            \begin{alignat}{2}
                g_{u,i,m} \rightarrow g_{u,i+1,m+d(\{u,u\})}
                &\qquad&
                \begin{aligned}[t]
                    &\forall u \in V,\; i \in [1,w(u)-1],\\
                    &m \in [LB_{u,i}, UB_{u,i}]
                \end{aligned}
            \end{alignat}

    \item \textit{Inter-vertex color distance constraint:}
            \begin{alignat}{2}
                x_{u,i,m} \rightarrow (g_{v,j,m+d(\{u,v\})} \vee \neg g_{v,j,m - d(\{u,v\})+1})
                &\qquad&
                \begin{aligned}[t]
                    &\forall \{u,v\} \in E, u \neq v,\\
                    &i \in [1, w(u)], j \in [1, w(v)],\\
                    &m \in [LB_{u,i}, UB_{u,i}]
                \end{aligned}
            \end{alignat}
\end{itemize}
 \textit{(Non-existent variables are considered as true or false constants accordingly and do not appear in the formula).}

Explicit SAT clauses for the \textit{Exactly-one value for the $i$-th color} constraint are unnecessary in this formulation. Specifically, linkage and monotonicity constraints ensure that the sequence $g_{u,i,m}$ transitions from true to false exactly once, forcing exactly one $x_{u,i,m}$ to evaluate to true for each index $i$. A formal proof for the order-based encoding is provided in Appendix \ref{apx:order-based-correctness}.

\subsection{Complexity Analysis} \label{sec:complexity}
In this section, we analyze the spatial complexity of the proposed SAT encodings by comparing the asymptotic number of variables and clauses they generate. To facilitate this comparison, let $W = \max_{u \in V} w(u)$ denote the maximum weight among all vertices in the graph. 

Table \ref{tab:complexity} provides a theoretical summary of the encoding sizes for the FE, SBE, and OBE. It is important to note that both the FE and SBE formulations utilize Sinz's sequential counter encoding \cite{sinz2005towards}.

\begin{table}[htb]
\centering
\caption{Comparison of encoding sizes}
\label{tab:complexity}
\begin{tabular}{lll} 
\hline
\multicolumn{1}{c}{\textbf{Encoding}} & \multicolumn{1}{c}{\textbf{Variables}} & \multicolumn{1}{c}{\textbf{Clauses}}                \\ 
\hline
FE                                     & $\mathcal{O}(|V| \cdot k^2)$            & $\mathcal{O}((|V|+|E|) \cdot k^2)$  \\
SBE                                    & $\mathcal{O}(|V| \cdot W \cdot k)$      & $\mathcal{O}(|E| \cdot W^2 \cdot k^2)$               \\
OBE                                    & $\mathcal{O}(|V| \cdot W \cdot k)$      & $\mathcal{O}(|E| \cdot W^2 \cdot k)$                 \\
\hline
\end{tabular}
\end{table}

\subsection{Finding the optimal solution} \label{sec:optimal-solution}

The SAT formulation described in the preceding sections models the BMCP as a decision problem: determining whether a valid coloring exists for a given fixed span. To solve the optimization variant of the BMCP—finding the minimum possible span $k^*$—we must systematically evaluate different span values. Rather than repeatedly generating and solving independent SAT instances from scratch, which discards valuable learned clauses and incurs significant encoding overhead, we employ an incremental SAT-solving strategy.

Our optimization process utilizes a top-down linear search. As established in Section \ref{subsec:bounds}, the upper bound algorithm guarantees an initial valid span $UB$. Therefore, the solver is instantiated exactly once with the complete encoding for the span value $UB$. To progressively search for a better span, we iteratively restrict the search space by forbidding the maximum currently allowed color. 

To evaluate if the graph can be colored with a span of $k-1$ (given that $k$ is currently valid), we do not need to alter the existing constraints. Instead, we dynamically inject unit clauses that strictly prohibit any vertex from being assigned the color value $k$. Depending on the chosen encoding, we append the corresponding incremental formula to the solver:

\begin{itemize}
    \item \textit{Flat encoding:}
        \begin{equation} \label{eq:incremental-flat}
            \neg x_{u,k} \qquad \forall u \in V
        \end{equation}
        
    \item \textit{Slot-based encoding:}
        \begin{equation} \label{eq:incremental-slot}
            \neg x_{u,i,k} \qquad \forall u \in V, i \in [1, w(u)] \quad \text{s.t.} \quad LB_{u,i} \leq k \leq UB_{u,i}
        \end{equation}

    \item \textit{Order-based encoding:}
        \begin{equation} \label{eq:incremental-order}
            \neg g_{u,i,k} \qquad \forall u \in V, i \in [1, w(u)] \quad \text{s.t.} \quad LB_{u,i} \leq k \leq UB_{u,i}
        \end{equation}
\end{itemize}

By adding these unit clauses, the SAT solver structurally invalidates any assignment relying on the color $k$, effectively forcing the target span to $k-1$. This process is repeated, decrementing the target span one by one, until the solver returns an \texttt{UNSAT} status. The optimal span $k^*$ is then strictly the last value for which a \texttt{SAT} status was achieved. 

The complete procedure is formalized in Algorithm \ref{alg:optimal-solving}.

\begin{algorithm}[htb]
\caption{Incremental Optimal Solving for BMCP}
\label{alg:optimal-solving}
\begin{algorithmic}[1]
\Require Graph $G=(V,E)$, initial valid upper bound $UB$ and lower bound $LB$
\Ensure Optimal span $k^*$
\State Encode complete BMCP instance for span $UB$
\State $k \gets UB$

\While{$k > LB$}
    \State \text{Add unit clauses forbidding color } $k$ \text{ (i.e., } $\neg x_{u,k}$, $\neg x_{u,i,k}$, \text{or } $\neg g_{u,i,k}$\text{)}
    \State $result \gets \text{SATSolve()}$
    \If{$result = \texttt{UNSAT}$}
        \State \Return $k$ \Comment{The last valid span before the problem became infeasible}
    \EndIf
    \State $k \gets k - 1$ \Comment{Prepare to test the next smaller span}
\EndWhile

\State \Return $k$ \Comment{If loop finishes, the optimal solution is exactly LB}
\end{algorithmic}
\end{algorithm}

\section{Experimental Evaluation}
\label{sec:experiments}

\subsection{Dataset}
\label{subsec:dataset}
We evaluate the proposed Order-Based Encoding (OBE) using 53 established benchmark instances from prior BMCP literature, divided into two distinct suites.

The first suite comprises 33 geometric instances, known as the GEOM set, introduced by Trick \cite{trick2002}. Vertices are distributed across a $10,000 \times 10,000$ grid, with edges connecting nodes within a predefined distance threshold; edge weights are inversely proportional to this distance. This set includes sparse (GEOMn) and dense (GEOMna, GEOMnb) configurations. Vertex weights are uniformly distributed in $[1, 10]$ for GEOMn and GEOMna, and $[1, 3]$ for GEOMnb.

The second suite contains the remaining 20 instances, derived from the classic Philadelphia (21 vertices) and Helsinki (25 vertices) instances, originally formulated for the Minimum Span Channel Assignment Problem (MS-CAP) \cite{chakraborty2001efficient}, plus a 55-vertex artificial extension. These instances directly model the physical topology and interference constraints of real-world cellular telecommunication networks.

\subsection{Experimental Setup}
\label{subsec:experimental-setup}
All local experiments were conducted on an Intel Core i5-12500H with 16 GB RAM running Ubuntu 24.04 LTS. A 3600-second time limit was applied to these experiments, covering both the encoding and search phases. To isolate encoding effectiveness, all SAT-based methods utilized CaDiCaL (v3.0.0) \cite{biere2023cadical_vivinst}. Additionally, all local experiments used identical initial lower and upper bounds (see Section \ref{subsec:bounds}). The baseline FE and SBE formulations, alongside our main OBE model, were implemented in C++. Specifically, the FE and SBE method utilizes the sequential counter encoding \cite{sinz2005towards} to enforce its cardinality constraints, whereas the OBE method handles these inherently. We conducted two primary sets of comparisons:

First, we compared OBE against FE, SBE, and two SAT models (POP-S-B and POPH-S-B) by Faber et al. \cite{faber2024}. Because the latter models only apply to the BCP, BMCP instances were first converted into equivalent BCP instances. We executed POP-S-B and POPH-S-B using the authors' original Python 3.14 source code\footnote{\url{https://github.com/s6dafabe/popsatgcpbcp}}.

Second, OBE was compared against exact models, heuristics, and metaheuristics. We re-implemented the CP and IP models proposed by Dias et al. \cite{dias2016constraint,dias2021}, evaluating them with CPLEX v22.2\footnote{\url{https://www.ibm.com/products/ilog-cplex-optimization-studio}} and Gurobi v13.0.2\footnote{\url{https://www.gurobi.com/}}. For the LPR heuristic \cite{lai2016} and VNS metaheuristic \cite{matic2017}, we cite the authors' published computational results directly; consequently, these were the only comparisons not bound by our shared hardware or the 3600-second time limit.

\subsection{Metrics}
\label{subsec:metrics}
To comprehensively evaluate and compare the performance of the various approaches, we defined specific evaluation metrics tailored to the two primary sets of comparisons outlined in Section \ref{subsec:experimental-setup}. Two main metrics were consistently applied across all comparisons: the best span value obtained ($k^*$) and the total computational time required to find or report the solution. To ensure statistical rigor, the evaluated methods were subsequently ranked using the Friedman test \cite{friedman1937use}, with the analysis conducted exclusively based on the best span value.

In addition to the core metrics, we evaluated the number of instances where the optimal span was successfully found and proven. This metric was applied to the entirety of the first set of comparisons (evaluating OBE against FE, SBE, POP-S-B, and POPH-S-B) as well as to the exact models evaluated in the second set. This metric was omitted for the LPR heuristic and VNS metaheuristic portions of the second comparison set, as approximate methods cannot definitively prove solution optimality.

\subsection{Results and Analysis}
\label{subsec:results-analysis}
In the following tables, underlined and boldface values indicate the best span ($k^*$) and proven optimal solutions, respectively. The rows \#OPTIMAL and \#BEST denote the total counts of optimal and best-known spans, while \#AVG. Rank reports the average Friedman test ranking. The \textit{Gap} column shows the difference between the OBE model's span and the overall minimum $k^*$. Finally, TO, MO, and a hyphen (-) represent Time Out (3600s), Memory Out, and failure to find a feasible solution, respectively.

\subsubsection{Comparison of OBE with Alternative SAT Encodings}

The computational results evaluating the efficacy of the proposed Order-Based Encoding (OBE) against alternative SAT encodings, specifically SBE, FE, POP-S-B, and POPH-S-B \cite{faber2024}, reveal a substantial performance paradigm shift in both solving capacity and solution quality. Tables \ref{tab:compare_sat_geom} and \ref{tab:compare_sat_mscap} detail this comparison across the GEOM and MS-CAP benchmark datasets, respectively.

\begin{table}[htpb]
\centering
\caption{Performance comparison of the OBE model with other SAT approaches on GEOM instances.}
\label{tab:compare_sat_geom}
\begin{adjustbox}{width=\linewidth}
\begin{tabular}{|l|l|l|l|l|l|l|l|l|l|l|l|} 
\hline
\multicolumn{1}{|c|}{\multirow{2}{*}{Instance}} & \multicolumn{2}{c|}{OBE}                            & \multicolumn{2}{c|}{SBE}                            & \multicolumn{2}{c|}{FE}                             & \multicolumn{2}{c|}{POP-S-B \cite{faber2024}}                        & \multicolumn{2}{c|}{POPH-S-B \cite{faber2024}}                       & \multicolumn{1}{c|}{\multirow{2}{*}{Gap}}  \\ 
\cline{2-11}
\multicolumn{1}{|c|}{}                          & \multicolumn{1}{c|}{$k^*$} & \multicolumn{1}{c|}{t(s)} & \multicolumn{1}{c|}{$k^*$} & \multicolumn{1}{c|}{t(s)} & \multicolumn{1}{c|}{$k^*$} & \multicolumn{1}{c|}{t(s)} & \multicolumn{1}{c|}{$k^*$} & \multicolumn{1}{c|}{t(s)} & \multicolumn{1}{c|}{$k^*$} & \multicolumn{1}{c|}{t(s)} & \multicolumn{1}{c|}{}                      \\ 
\hline
GEOM20                                          & \underline{\textbf{149}}    & 4.10                      & \underline{149}             & TO                        & \textbf{\underline{149}}    & 731.76                    & \underline{149}             & TO                        & \underline{149}             & TO                        & 0                                          \\ 
\hline
GEOM20a                                         & \underline{\textbf{169}}    & 2.13                      & \underline{169}             & TO                        & \textbf{\underline{169}}    & 25.26                     & \underline{169}             & TO                        & \underline{169}             & TO                        & 0                                          \\ 
\hline
GEOM20b                                         & \underline{\textbf{44}}     & 0.03                      & \textbf{\underline{44}}     & 0.47                      & \textbf{\underline{44}}     & 0.07                      & \textbf{\underline{44}}     & 0.24                      & \textbf{\underline{44}}     & 0.37                      & 0                                          \\ 
\hline
GEOM30                                          & \underline{\textbf{160}}    & 213.36                    & 164                     & TO                        & 161             & TO                        & \underline{160}             & TO                        & \underline{160}             & TO                        & 0                                          \\ 
\hline
GEOM30a                                         & \underline{\textbf{209}}    & 1,240.10                  & 218                     & TO                        & 210             & TO                        & 214                     & TO                        & 217                     & TO                        & 0                                          \\ 
\hline
GEOM30b                                         & \underline{\textbf{77}}     & 0.40                      & \textbf{\underline{77}}     & 146.48                    & \textbf{\underline{77}}     & 27.09                    & \textbf{\underline{77}}     & 3.03                      & \textbf{\underline{77}}     & 4.37                      & 0                                          \\ 
\hline
GEOM40                                          & \underline{\textbf{167}}    & 5.66                      & 171                     & TO                        & \textbf{\underline{167}}    & 41.87                    & \underline{167}             & TO                        & \underline{167}             & TO                        & 0                                          \\ 
\hline
GEOM40a                                         & \underline{\textbf{213}}    & 13.15                     & 267                     & TO                        & \textbf{\underline{213}}    & 339.3                    & \underline{213}             & TO                        & \underline{213}             & TO                        & 0                                          \\ 
\hline
GEOM40b                                         & \underline{\textbf{74}}     & 1.50                      & \textbf{\underline{74}}     & 351.17                    & \textbf{\underline{74}}     & 1400.41                     & \textbf{\underline{74}}     & 296.77                    & \textbf{\underline{74}}     & 385.48                    & 0                                          \\ 
\hline
GEOM50                                          & \underline{\textbf{224}}    & 1,193.77                  & 272                     & TO                        & 225             & TO                        & \underline{224}             & TO                        & \underline{224}             & TO                        & 0                                          \\ 
\hline
GEOM50a                                         & \underline{311}             & TO                        & 404                     & TO                        & 318                     & TO                        & 326                     & TO                        & 326                     & TO                        & 0                                          \\ 
\hline
GEOM50b                                         & \underline{\textbf{83}}     & 22.37                     & 86                      & TO                        & 84              & TO                        & \underline{83}              & TO                        & \underline{83}              & TO                        & 0                                          \\ 
\hline
GEOM60                                          & \underline{\textbf{258}}    & 1,013.92                  & 305                     & TO                        & 259             & TO                        & \underline{258}             & TO                        & \underline{258}             & TO                        & 0                                          \\ 
\hline
GEOM60a                                         & \underline{357}             & TO                        & 469                     & TO                        & 363                     & TO                        & 368                     & TO                        & 370                     & TO                        & 0                                          \\ 
\hline
GEOM60b                                         & \underline{\textbf{113}}    & 858.64                    & 114                     & TO                        & 114             & TO                        & 114                     & TO                        & \underline{113}             & TO                        & 0                                          \\ 
\hline
GEOM70                                          & \underline{266}             & TO                        & 301                     & TO                        & 271                     & TO                        & 276                     & TO                        & 278                     & TO                        & 0                                          \\ 
\hline
GEOM70a                                         & \underline{472}             & TO                        & 578                     & TO                        & 477                     & TO                        & 483                     & TO                        & 486                     & TO                        & 0                                          \\ 
\hline
GEOM70b                                         & \underline{\textbf{115}}    & 3,025.24                  & 124                     & TO                        & 119                     & TO                        & 116                     & TO                        & 116                     & TO                        & 0                                          \\ 
\hline
GEOM80                                          & \underline{383}             & TO                        & 511                     & TO                        & 388                     & TO                        & 400                     & TO                        & 404                     & TO                        & 0                                          \\ 
\hline
GEOM80a                                         & \underline{364}             & TO                        & 479                     & TO                        & 373                     & TO                        & 386                     & TO                        & 388                     & TO                        & 0                                          \\ 
\hline
GEOM80b                                         & \underline{138}             & TO                        & 145                     & TO                        & 141                     & TO                        & 139                     & TO                        & \underline{138}             & TO                        & 0                                          \\ 
\hline
GEOM90                                          & \underline{328}             & TO                        & 422                     & TO                        & 334                     & TO                        & 337                     & TO                        & 337                     & TO                        & 0                                          \\ 
\hline
GEOM90a                                         & \underline{375}             & TO                        & 468                     & TO                        & 379                     & TO                        & 385                     & TO                        & 387                     & TO                        & 0                                          \\ 
\hline
GEOM90b                                         & \underline{142}             & TO                        & 165                     & TO                        & 150                     & TO                        & 145                     & TO                        & 147                     & TO                        & 0                                          \\ 
\hline
GEOM100                                         & \underline{411}             & TO                        & 494                     & TO                        & 420                     & TO                        & 421                     & TO                        & 424                     & TO                        & 0                                          \\ 
\hline
GEOM100a                                        & \underline{447}             & TO                        & -                       & MO                        & 464                     & TO                        & 472                     & TO                        & 487                     & TO                        & 0                                          \\ 
\hline
GEOM100b                                        & \underline{152}             & TO                        & 223                     & TO                        & 156                     & TO                        & 156                     & TO                        & 156                     & TO                        & 0                                          \\ 
\hline
GEOM110                                         & \underline{379}             & TO                        & -                       & MO                        & 386                     & TO                        & 396                     & TO                        & 398                     & TO                        & 0                                          \\ 
\hline
GEOM110a                                        & \underline{494}             & TO                        & -                       & MO                        & 507                     & TO                        & 524                     & TO                        & 536                     & TO                        & 0                                          \\ 
\hline
GEOM110b                                        & \underline{202}             & TO                        & 261                     & TO                        & 205                     & TO                        & 204                     & TO                        & 203                     & TO                        & 0                                          \\ 
\hline
GEOM120                                         & \underline{399}             & TO                        & -                       & MO                        & 414                     & TO                        & 418                     & TO                        & 425                     & TO                        & 0                                          \\ 
\hline
GEOM120a                                        & \underline{554}             & TO                        & -                       & MO                        & 569                     & TO                        & 603                     & TO                        & 620                     & TO                        & 0                                          \\ 
\hline
GEOM120b                                        & \underline{189}             & TO                        & 260                     & TO                        & 194                     & TO                        & 191                     & TO                        & 191                     & TO                        & 0                                          \\ 
\hline
\multicolumn{1}{|c|}{\#OPTIMAL}                 & \multicolumn{2}{c|}{14}                             & \multicolumn{2}{c|}{3}                              & \multicolumn{2}{c|}{7}                              & \multicolumn{2}{c|}{3}                              & \multicolumn{2}{c|}{3}                              & \multicolumn{1}{c|}{-}                     \\ 
\hline
\multicolumn{1}{|c|}{\#BEST}                    & \multicolumn{2}{c|}{33}                             & \multicolumn{2}{c|}{5}                              & \multicolumn{2}{c|}{7}                             & \multicolumn{2}{c|}{11}                             & \multicolumn{2}{c|}{13}                             & \multicolumn{1}{c|}{-}                     \\ 
\hline
\multicolumn{1}{|c|}{\#AVG. Rank}               & \multicolumn{2}{c|}{1.55}                           & \multicolumn{2}{c|}{4.67}                           & \multicolumn{2}{c|}{2.82}                           & \multicolumn{2}{c|}{2.85}                           & \multicolumn{2}{c|}{3.12}                           & \multicolumn{1}{c|}{-}                     \\
\hline
\end{tabular}
\end{adjustbox}
\end{table}

\begin{table}[htpb]
\centering
\caption{Performance comparison of the OBE model with other SAT approaches on MS-CAP instances.}
\label{tab:compare_sat_mscap}
\begin{adjustbox}{width=\linewidth}
\begin{tabular}{|l|l|l|l|l|l|l|l|l|l|l|l|l|} 
\hline
\multicolumn{1}{|c|}{\multirow{2}{*}{\begin{tabular}[c]{@{}c@{}}Const. \\Matr.\end{tabular}}} & \multicolumn{1}{c|}{\multirow{2}{*}{\begin{tabular}[c]{@{}c@{}}Demd. \\Vect.\end{tabular}}} & \multicolumn{2}{c|}{OBE}                       & \multicolumn{2}{c|}{SBE}                            & \multicolumn{2}{c|}{FE}                             & \multicolumn{2}{c|}{POP-S-B \cite{faber2024}}                        & \multicolumn{2}{c|}{POPH-S-B \cite{faber2024}}                       & \multicolumn{1}{c|}{\multirow{2}{*}{Gap}}  \\ 
\cline{3-12}
\multicolumn{1}{|c|}{}                                                                        & \multicolumn{1}{c|}{}                                                                       & \multicolumn{1}{c|}{$k^*$} & \multicolumn{1}{c|}{t(s)} & \multicolumn{1}{c|}{$k^*$} & \multicolumn{1}{c|}{t(s)} & \multicolumn{1}{c|}{$k^*$} & \multicolumn{1}{c|}{t(s)} & \multicolumn{1}{c|}{$k^*$} & \multicolumn{1}{c|}{t(s)} & \multicolumn{1}{c|}{$k^*$} & \multicolumn{1}{c|}{t(s)} & \multicolumn{1}{c|}{}                      \\ 
\hline
$C^1_{21}$                                                                                    & $D^1_{21}$                                                                                  & \textbf{\underline{533}}    & 63.87                     & \textbf{\underline{533}}    & 66.94                     & \textbf{\underline{533}}    & 10.54                     & \textbf{\underline{533}}    & 44.80                     & -                       & TO                        & 0                                          \\ 
\hline
$C^1_{21}$                                                                                    & $D^2_{21}$                                                                                  & \textbf{\underline{309}}    & 21.82                     & \textbf{\underline{309}}    & 14.29                     & \textbf{\underline{309}}    & 5.85                      & \textbf{\underline{309}}    & 4.74                      & \underline{\textbf{309}}    & 3.74                      & 0                                          \\ 
\hline
$C^2_{21}$                                                                                    & $D^1_{21}$                                                                                  & \textbf{\underline{533}}    & 54.40                     & \textbf{\underline{533}}    & 164.33                    & \textbf{\underline{533}}    & 30.12                     & 535                     & TO                        & -                       & TO                        & 0                                          \\ 
\hline
$C^2_{21}$                                                                                    & $D^2_{21}$                                                                                  & \textbf{\underline{309}}    & 36.66                     & \textbf{\underline{309}}    & 145.31                    & \textbf{\underline{309}}    & 124.28                    & 312                     & TO                        & 317                     & TO                        & 0                                          \\ 
\hline
$C^3_{21}$                                                                                    & $D^1_{21}$                                                                                  & \textbf{\underline{457}}    & 55.33                     & \textbf{\underline{457}}    & 48.86                     & \textbf{\underline{457}}    & 6.65                     & \textbf{\underline{457}}    & 1,346.18                  & -                       & TO                        & 0                                          \\ 
\hline
$C^3_{21}$                                                                                    & $D^2_{21}$                                                                                  & \textbf{\underline{265}}    & 15.86                     & \textbf{\underline{265}}    & 10.42                     & \textbf{\underline{265}}    & 6.54                     & \textbf{\underline{265}}    & 45.66                     & \underline{\textbf{265}}    & 1,224.81                  & 0                                          \\ 
\hline
$C^4_{21}$                                                                                    & $D^1_{21}$                                                                                  & \textbf{\underline{457}}    & 74.31                     & \textbf{\underline{457}}    & 906.15                    & \textbf{\underline{457}}    & 271.2                  & 485                     & TO                        & 489                     & TO                        & 0                                          \\ 
\hline
$C^4_{21}$                                                                                    & $D^2_{21}$                                                                                  & \textbf{\underline{265}}    & 39.68                     & \textbf{\underline{265}}    & 447.98                    & \textbf{\underline{265}}    & 104.58                    & 279                     & TO                        & 282                     & TO                        & 0                                          \\ 
\hline
$C^5_{21}$                                                                                    & $D^1_{21}$                                                                                  & \textbf{\underline{381}}    & 44.67                     & \textbf{\underline{381}}    & 35.62                     & \textbf{\underline{381}}    & 4.28                      & -                       & TO                        & -                       & TO                        & 0                                          \\ 
\hline
$C^5_{21}$                                                                                    & $D^2_{21}$                                                                                  & \textbf{\underline{221}}    & 19.21                     & \textbf{\underline{221}}    & 35.16                     & \textbf{\underline{221}}    & 34.78                     & \textbf{\underline{221}}    & 739.30                    & \underline{\textbf{221}}    & 2,296.05                  & 0                                          \\ 
\hline
$C^6_{21}$                                                                                    & $D^1_{21}$                                                                                  & \underline{427}             & TO                        & 461                     & TO                        & 428             & TO                        & 474                     & TO                        & 469                     & TO                        & 0                                          \\ 
\hline
$C^6_{21}$                                                                                    & $D^2_{21}$                                                                                  & \underline{253}             & TO                        & 283                     & TO                        & 254             & TO                        & 263                     & TO                        & 265                     & TO                        & 0                                          \\ 
\hline
$C^7_{21}$                                                                                    & $D^1_{21}$                                                                                  & \textbf{\underline{305}}    & 30.70                     & \textbf{\underline{305}}    & 24.46                     & \textbf{\underline{305}}    & 7.46                     & 307                     & TO                        & 307                     & TO                        & 0                                          \\ 
\hline
$C^7_{21}$                                                                                    & $D^2_{21}$                                                                                  & \textbf{\underline{180}}    & 50.79                     & \textbf{\underline{180}}    & 236.95                    & \textbf{\underline{180}}    & 63.67                    & 181                     & TO                        & 183                     & TO                        & 0                                          \\ 
\hline
$C^8_{21}$                                                                                    & $D^1_{21}$                                                                                  & 432             & TO                        & 534                     & TO                        & \underline{428}                     & TO                        & 464                     & TO                        & 461                     & TO                        & 4                                          \\ 
\hline
$C^8_{21}$                                                                                    & $D^2_{21}$                                                                                  & 256             & TO                        & 338                     & TO                        & \underline{255}             & TO                        & 264                     & TO                        & 263                     & TO                        & 1                                          \\ 
\hline
$C^1_{25}$                                                                                    & $D^3_{25}$                                                                                  & \textbf{\underline{73}}     & 1.89                      & \textbf{\underline{73}}     & 1.05                      & \textbf{\underline{73}}     & 0.12                      & \textbf{\underline{73}}     & 0.12                      & \underline{\textbf{73}}     & 0.12                      & 0                                          \\ 
\hline
$C^1_{25}$                                                                                    & $D^4_{25}$                                                                                  & \textbf{\underline{200}}    & 44.13                     & \textbf{\underline{200}}    & 131.86                    & \textbf{\underline{200}}    & 1.72                      & \textbf{\underline{200}}    & 14.93                     & \underline{\textbf{200}}    & 21.38                     & 0                                          \\ 
\hline
$C^1_{55}$                                                                                    & $D^5_{55}$                                                                                  & \textbf{\underline{309}}    & 58.49                     & \textbf{\underline{309}}    & 36.83                     & \textbf{\underline{309}}    & 31.52                     & -                       & TO                        & -                       & TO                        & 0                                          \\ 
\hline
$C^1_{55}$                                                                                    & $D^6_{55}$                                                                                  & \textbf{\underline{71}}     & 2.26                      & \textbf{\underline{71}}     & 5.50                      & \textbf{\underline{71}}     & 2.78                     & \textbf{\underline{71}}     & 54.74                     & \underline{\textbf{71}}     & 182.59                    & 0                                          \\ 
\hline
\multicolumn{2}{|c|}{\#OPTIMAL}                                                                                                                                                             & \multicolumn{2}{c|}{16}                             & \multicolumn{2}{c|}{16}                             & \multicolumn{2}{c|}{16}                             & \multicolumn{2}{c|}{8}                              & \multicolumn{2}{c|}{6}                              & \multicolumn{1}{c|}{-}                     \\ 
\hline
\multicolumn{2}{|c|}{\#BEST}                                                                                                                                                                & \multicolumn{2}{c|}{18}                             & \multicolumn{2}{c|}{16}                             & \multicolumn{2}{c|}{18}                             & \multicolumn{2}{c|}{8}                              & \multicolumn{2}{c|}{6}                              & \multicolumn{1}{c|}{-}                     \\ 
\hline
\multicolumn{2}{|c|}{\#AVG. Rank}                                                                                                                                                           & \multicolumn{2}{c|}{2.25}                           & \multicolumn{2}{c|}{2.85}                           & \multicolumn{2}{c|}{2.25}                           & \multicolumn{2}{c|}{3.63}                           & \multicolumn{2}{c|}{4.03}                           & \multicolumn{1}{c|}{-}                     \\
\hline
\end{tabular}
\end{adjustbox}
\end{table}

Regarding the capacity to prove optimality, the proposed OBE formulation demonstrates a superiority over existing models. Most strikingly, on the GEOM dataset, OBE exactly doubles the number of proven optimal solutions compared to FE, solving 14 instances to optimality versus FE's 7. Furthermore, OBE is the only SAT encoding capable of proving optimality on instances such as GEOM30, GEOM50, GEOM60, and GEOM70b, where all other models prematurely exhausted the computational time limit. While the performance gap narrows on the MS-CAP dataset, with OBE, SBE, and FE each proving optimality for 16 out of 20 configurations, OBE maintains its top-tier status while simultaneously exposing the limitations of the specialized POP-S-B and POPH-S-B models, which managed to prove only 8 and 6 optimal solutions, respectively. This disparity indicates that OBE's underlying logical structure significantly accelerates the solver's ability to prune the search space and verify lower bounds.

Beyond the verification of exact optima, the true practical significance of the OBE formulation lies in its robustness when processing difficult instances. In this regard, OBE is exceptionally reliable. Across the entirety of the 53 instances spanning both datasets, the calculated gap between the span identified by OBE and the minimum overall $k^*$ is zero for 51 instances. Consequently, even when interrupted by a timeout, OBE consistently returns a solution that is highly competitive with, and frequently superior to, the best solutions generated by alternative SAT encodings, with FE being the only model to find marginally tighter bounds on two specific MS-CAP configurations.

The interpretation of the bounds further underscores OBE's dominance. OBE successfully secured the best-known span in the vast majority of the evaluated configurations (33 out of 33 in GEOM, and 18 out of 20 in MS-CAP). In stark contrast, the search trajectories of alternative encodings often stagnated far from the optimal region. For example, on the GEOM dataset, the closest competitor in terms of solution quality (POPH-S-B) could only reach the best bound in approximately 39\% of the instances (13 out of 33), while FE reached it in only about 21\% (7 out of 33). On challenging MS-CAP instances utilizing the $C^6_{21}$ matrices, alternative models struggled to find high-quality spans, whereas OBE reported tighter bounds. Conversely, FE demonstrated a slight advantage on instances utilizing the $C^8_{21}$ matrices, where it managed to surpass OBE. This empirical evidence signifies that OBE not only elevates the ceiling of SAT-based optimization by solving more instances to proven optimality, but also fundamentally raises the floor by consistently providing highly competitive, near-optimal bounds when exact verification remains computationally intractable.

\begin{landscape}
\begin{table}
\caption{Performance comparison of the OBE model with exact models, heuristics, and metaheuristics on GEOM instances.}
\label{tab:compare_others_geom}
\centering
\begin{adjustbox}{width=\linewidth}
\begin{tabular}{|l|l|l|l|l|l|l|l|l|l|l|l|l|l|l|l|l|l|} 
\hline
\multicolumn{1}{|c|}{\multirow{3}{*}{Instance}} & \multicolumn{8}{c|}{Our implementation (TO = 3600s)}                                                                                                                                                                  & \multicolumn{4}{c|}{Dias et al. \cite{dias2016constraint,dias2021} (TO = 172,800s)}                                                          & \multicolumn{2}{c|}{Lai et al. \cite{lai2016}}                     & \multicolumn{2}{c|}{Mati\'c et al. \cite{matic2017}}                   & \multicolumn{1}{c|}{\multirow{3}{*}{Gap}}  \\ 
\cline{2-17}
\multicolumn{1}{|c|}{}                          & \multicolumn{2}{c|}{OBE}                            & \multicolumn{2}{c|}{CPLEX CP}                       & \multicolumn{2}{c|}{CPLEX ILP}                      & \multicolumn{2}{c|}{Gurobi}                         & \multicolumn{2}{c|}{CP}                             & \multicolumn{2}{c|}{IP}                             & \multicolumn{2}{c|}{LPR}                            & \multicolumn{2}{c|}{VNS}                            & \multicolumn{1}{c|}{}                      \\ 
\cline{2-17}
\multicolumn{1}{|c|}{}                          & \multicolumn{1}{c|}{$k^*$} & \multicolumn{1}{c|}{t(s)} & \multicolumn{1}{c|}{$k^*$} & \multicolumn{1}{c|}{t(s)} & \multicolumn{1}{c|}{$k^*$} & \multicolumn{1}{c|}{t(s)} & \multicolumn{1}{c|}{$k^*$} & \multicolumn{1}{c|}{t(s)} & \multicolumn{1}{c|}{$k^*$} & \multicolumn{1}{c|}{t(s)} & \multicolumn{1}{c|}{$k^*$} & \multicolumn{1}{c|}{t(s)} & \multicolumn{1}{c|}{$k^*$} & \multicolumn{1}{c|}{t(s)} & \multicolumn{1}{c|}{$k^*$} & \multicolumn{1}{c|}{t(s)} & \multicolumn{1}{c|}{}                      \\ 
\hline
GEOM20                                          & \textbf{\underline{149}}    & 4.10                      & \underline{149}             & TO                        & \textbf{\underline{149}}    & 137.4                     & \textbf{\underline{149}}    & 2.59                      & \underline{149}             & TO                        & \textbf{\underline{149}}    & 15.17                     & \underline{149}             & 0.00                      & \underline{149}             & 54.27                     & 0                                          \\ 
\hline
GEOM20a                                         & \textbf{\underline{169}}    & 2.13                      & \underline{169}             & TO                        & \textbf{\underline{169}}    & 4.72                      & \textbf{\underline{169}}    & 1.49                      & \underline{169}             & TO                        & \textbf{\underline{169}}    & 18.49                     & \underline{169}             & 0.10                      & \underline{169}             & 3,289                     & 0                                          \\ 
\hline
GEOM20b                                         & \textbf{\underline{44}}     & 0.03                      & \textbf{\underline{44}}     & 0.01                      & \textbf{\underline{44}}     & 0.70                      & \textbf{\underline{44}}     & 0.65                      & \textbf{\underline{44}}     & 476.9                     & \textbf{\underline{44}}     & 1.58                      & \underline{44}              & 0.00                      & \underline{44}              & 0.02                      & 0                                          \\ 
\hline
GEOM30                                          & \textbf{\underline{160}}    & 213.4                     & \underline{160}             & TO                        & \underline{160}             & TO                        & \textbf{\underline{160}}    & 1,656                     & \underline{160}             & TO                        & \underline{160}             & TO                        & \underline{160}             & 0.00                      & \underline{160}             & 5.70                      & 0                                          \\ 
\hline
GEOM30a                                         & \textbf{\underline{209}}    & 1,240                     & 210                     & TO                        & 215                     & TO                        & 210                     & TO                        & 215                     & TO                        & 211                     & TO                        & \underline{209}             & 1.40                      & \underline{209}             & 4,124                     & 0                                          \\ 
\hline
GEOM30b                                         & \textbf{\underline{77}}     & 0.40                      & \textbf{\underline{77}}     & 46.91                     & \textbf{\underline{77}}     & 0.85                      & \textbf{\underline{77}}     & 0.68                      & \underline{77}              & TO                        & \textbf{\underline{77}}     & 41.87                     & \underline{77}              & 0.00                      & \underline{77}              & 1.29                      & 0                                          \\ 
\hline
GEOM40                                          & \textbf{\underline{167}}    & 5.66                      & \underline{167}             & TO                        & \textbf{\underline{167}}    & 18.47                     & \textbf{\underline{167}}    & 15.03                     & 168                     & TO                        & \textbf{\underline{167}}    & 1,192.28                  & \underline{167}             & 0.00                      & \underline{167}             & 2,107                     & 0                                          \\ 
\hline
GEOM40a                                         & \textbf{\underline{213}}    & 13.15                     & \underline{213}             & TO                        & -                       & TO                        & \textbf{\underline{213}}    & 153.1                     & 225                     & TO                        & \textbf{\underline{213}}    & 111,262                   & \underline{213}             & 6.20                      & \underline{213}             & 13,192                    & 0                                          \\ 
\hline
GEOM40b                                         & \textbf{\underline{74}}     & 1.50                      & \textbf{\underline{74}}     & 182.6                     & \textbf{\underline{74}}     & 1,053                     & \textbf{\underline{74}}     & 62.69                     & \underline{74}              & TO                        & \textbf{\underline{74}}     & 17,028                    & \underline{74}              & 0.10                      & \underline{74}              & 1,822                     & 0                                          \\ 
\hline
GEOM50                                          & \textbf{\underline{224}}    & 1,194                     & \underline{224}             & TO                        & 226                     & TO                        & \textbf{\underline{224}}    & 17.19                     & 226                     & TO                        & \textbf{\underline{224}}    & 52,451                    & \underline{224}             & 0.10                      & \underline{224}             & 1,671                     & 0                                          \\ 
\hline
GEOM50a                                         & \underline{311}             & TO                        & 315                     & TO                        & -                       & TO                        & 337                     & TO                        & 332                     & TO                        & 361                     & TO                        & \underline{311}             & 3,560                     & \underline{311}             & 21,686                    & 0                                          \\ 
\hline
GEOM50b                                         & \textbf{\underline{83}}     & 22.37                     & \underline{83}              & TO                        & 111                     & TO                        & 90                      & TO                        & 85                      & TO                        & 87                      & TO                        & \underline{83}              & 16.80                     & \underline{83}              & 14,261                    & 0                                          \\ 
\hline
GEOM60                                          & \textbf{\underline{258}}    & 1,014                     & \underline{258}             & TO                        & -                       & TO                        & \textbf{\underline{258}}    & 49.48                     & 259                     & TO                        & \textbf{\underline{258}}    & 156,988                   & \underline{258}             & 0.10                      & \underline{258}             & 5,780                     & 0                                          \\ 
\hline
GEOM60a                                         & 357                     & TO                        & 358                     & TO                        & 460                     & TO                        & 368                     & TO                        & 380                     & TO                        & 448                     & TO                        & \underline{353}             & 3,819                     & \underline{353}             & 23,989                    & 4                                          \\ 
\hline
GEOM60b                                         & \textbf{\underline{113}}    & 858.6                     & \underline{113}             & TO                        & -                       & TO                        & 122                     & TO                        & 117                     & TO                        & 125                     & TO                        & \underline{113}             & 710.30                    & 114                     & 16,381                    & 0                                          \\ 
\hline
GEOM70                                          & \underline{266}             & TO                        & 270                     & TO                        & 374                     & TO                        & 270                     & TO                        & 284                     & TO                        & 305                     & TO                        & \underline{266}             & 1,569                     & 267                     & 12,384                    & 0                                          \\ 
\hline
GEOM70a                                         & 472                     & TO                        & 467                     & TO                        & 577                     & TO                        & 472                     & TO                        & 483                     & TO                        & 578                     & TO                        & 465                     & 9,790                     & \underline{463}             & 20,686                    & 9                                          \\ 
\hline
GEOM70b                                         & \textbf{\underline{115}}    & 3,025                     & 117                     & TO                        & -                       & TO                        & 130                     & TO                        & 123                     & TO                        & 134                     & TO                        & \underline{115}             & 764.6                     & 116                     & 16,783                    & 0                                          \\ 
\hline
GEOM80                                          & 383                     & TO                        & 382                     & TO                        & -                       & TO                        & \textbf{\underline{379}}    & 511.7                     & 395                     & TO                        & 511                     & TO                        & \underline{379}             & 1,983                     & \underline{379}             & 16,539                    & 4                                          \\ 
\hline
GEOM80a                                         & 364                     & TO                        & 366                     & TO                        & 461                     & TO                        & 397                     & TO                        & 382                     & TO                        & 479                     & TO                        & \underline{352}             & 7,880                     & 355                     & 29,209                    & 12                                         \\ 
\hline
GEOM80b                                         & \underline{138}             & TO                        & 139                     & TO                        & 179                     & TO                        & 152                     & TO                        & 145                     & TO                        & 170                     & TO                        & \underline{138}             & 51.80                     & \underline{138}             & 13,704                    & 0                                          \\ 
\hline
GEOM90                                          & \underline{328}             & TO                        & 331                     & TO                        & 418                     & TO                        & \textbf{\underline{328}}    & 2,105                     & 342                     & TO                        & 423                     & TO                        & \underline{328}             & 2,168                     & 329                     & 18,761                    & 0                                          \\ 
\hline
GEOM90a                                         & 375                     & TO                        & 380                     & TO                        & -                       & TO                        & 407                     & TO                        & 392                     & TO                        & 452                     & TO                        & \underline{372}             & 3,912                     & 373                     & 24,087                    & 3                                          \\ 
\hline
GEOM90b                                         & 142                     & TO                        & 144                     & TO                        & 210                     & TO                        & 173                     & TO                        & 156                     & TO                        & 212                     & TO                        & \underline{140}             & 2,142                     & 142                     & 19,997                    & 2                                          \\ 
\hline
GEOM100                                         & 411                     & TO                        & 405                     & TO                        & -                       & TO                        & \underline{404}             & TO                        & 426                     & TO                        & 493                     & TO                        & \underline{404}             & 64.80                     & \underline{404}             & 14,817                    & 7                                          \\ 
\hline
GEOM100a                                        & 447                     & TO                        & 444                     & TO                        & -                       & TO                        & 532                     & TO                        & 465                     & TO                        & 596                     & TO                        & \underline{426}             & 13,060                    & 429                     & 35,663                    & 21                                         \\ 
\hline
GEOM100b                                        & 152                     & TO                        & 155                     & TO                        & -                       & TO                        & 198                     & TO                        & 169                     & TO                        & 220                     & TO                        & \underline{151}             & 3,179                     & 156                     & 24,776                    & 1                                          \\ 
\hline
GEOM110                                         & 379                     & TO                        & 378                     & TO                        & 483                     & TO                        & 450                     & TO                        & 399                     & TO                        & 500                     & TO                        & \underline{375}             & 3,510                     & \underline{375}             & 22,997                    & 4                                          \\ 
\hline
GEOM110a                                        & 494                     & TO                        & 492                     & TO                        & -                       & TO                        & -                       & TO                        & 527                     & TO                        & 610                     & TO                        & \underline{478}             & 10,979                    & 480                     & 46,955                    & 16                                         \\ 
\hline
GEOM110b                                        & 202                     & TO                        & 203                     & TO                        & -                       & TO                        & 218                     & TO                        & 207                     & TO                        & 250                     & TO                        & \underline{201}             & 1,920                     & 202                     & 22,077                    & 1                                          \\ 
\hline
GEOM120                                         & 399                     & TO                        & 399                     & TO                        & -                       & TO                        & 402                     & TO                        & 427                     & TO                        & 505                     & TO                        & \underline{396}             & 999.2                     & \underline{396}             & 17,919                    & 3                                          \\ 
\hline
GEOM120a                                        & 554                     & TO                        & 547                     & TO                        & -                       & TO                        & 606                     & TO                        & 585                     & TO                        & 641                     & TO                        & \underline{531}             & 24,880                    & 539                     & 59,717                    & 23                                         \\ 
\hline
GEOM120b                                        & 189                     & TO                        & 192                     & TO                        & -                       & TO                        & 226                     & TO                        & 202                     & TO                        & 247                     & TO                        & \underline{187}             & 1,837.4                   & 190                     & 22,836                    & 2                                          \\ 
\hline
\multicolumn{1}{|c|}{\#OPTIMAL}                 & \multicolumn{2}{c|}{14}                             & \multicolumn{2}{c|}{3}                              & \multicolumn{2}{c|}{6}                              & \multicolumn{2}{c|}{12}                             & \multicolumn{2}{c|}{1}                              & \multicolumn{2}{c|}{9}                              & \multicolumn{2}{c|}{-}                              & \multicolumn{2}{c|}{-}                              & \multicolumn{1}{c|}{-}                     \\ 
\hline
\multicolumn{1}{|c|}{\#BEST}                    & \multicolumn{2}{c|}{18}                             & \multicolumn{2}{c|}{12}                             & \multicolumn{2}{c|}{7}                              & \multicolumn{2}{c|}{13}                             & \multicolumn{2}{c|}{6}                              & \multicolumn{2}{c|}{10}                             & \multicolumn{2}{c|}{32}                             & \multicolumn{2}{c|}{20}                             & \multicolumn{1}{c|}{-}                     \\ 
\hline
\multicolumn{1}{|c|}{\#AVG. Rank}               & \multicolumn{2}{c|}{3.30}                           & \multicolumn{2}{c|}{3.80}                           & \multicolumn{2}{c|}{7.05}                           & \multicolumn{2}{c|}{4.95}                           & \multicolumn{2}{c|}{5.48}                           & \multicolumn{2}{c|}{6.20}                           & \multicolumn{2}{c|}{2.32}                           & \multicolumn{2}{c|}{2.89}                           & \multicolumn{1}{c|}{-}                     \\
\hline
\end{tabular}
\end{adjustbox}
\end{table}
\end{landscape}

\begin{landscape}
\begin{table}
\caption{Performance comparison of the OBE model with exact models on MS-CAP instances.}
\label{tab:compare_others_mscap}
\centering
\begin{adjustbox}{width=\linewidth}
\begin{tabular}{|l|l|l|l|l|l|l|l|l|l|l|l|l|l|l|} 
\hline
\multicolumn{1}{|c|}{\multirow{3}{*}{\begin{tabular}[c]{@{}c@{}}Const. \\Matr.\end{tabular}}} & \multicolumn{1}{c|}{\multirow{3}{*}{\begin{tabular}[c]{@{}c@{}}Demd. \\Vect.\end{tabular}}} & \multicolumn{8}{c|}{Our implementation (TO = 3600s)}                                                                                                                                                                  & \multicolumn{4}{c|}{Dias et al. \cite{dias2016constraint, dias2021} ~(TO = 172,800s)}                                                          & \multicolumn{1}{c|}{\multirow{3}{*}{Gap}}  \\ 
\cline{3-14}
\multicolumn{1}{|c|}{}                                                                        & \multicolumn{1}{c|}{}                                                                       & \multicolumn{2}{c|}{OBE}                            & \multicolumn{2}{c|}{CPLEX CP}                       & \multicolumn{2}{c|}{CPLEX ILP}                      & \multicolumn{2}{c|}{Gurobi}                         & \multicolumn{2}{c|}{CP}                             & \multicolumn{2}{c|}{IP}                             & \multicolumn{1}{c|}{}                      \\ 
\cline{3-14}
\multicolumn{1}{|c|}{}                                                                        & \multicolumn{1}{c|}{}                                                                       & \multicolumn{1}{c|}{$k^*$} & \multicolumn{1}{c|}{t(s)} & \multicolumn{1}{c|}{$k^*$} & \multicolumn{1}{c|}{t(s)} & \multicolumn{1}{c|}{$k^*$} & \multicolumn{1}{c|}{t(s)} & \multicolumn{1}{c|}{$k^*$} & \multicolumn{1}{c|}{t(s)} & \multicolumn{1}{c|}{$k^*$} & \multicolumn{1}{c|}{t(s)} & \multicolumn{1}{c|}{$k^*$} & \multicolumn{1}{c|}{t(s)} & \multicolumn{1}{c|}{}                      \\ 
\hline
$C^1_{21}$                                                                                    & $D^1_{21}$                                                                                  & \textbf{\underline{533}}    & 63.87                     & \textbf{\underline{533}}    & 1.16                      & \textbf{\underline{533}}    & 1.07                      & \textbf{\underline{533}}    & 0.54                      & \textbf{\underline{533}}    & 4.2                       & \textbf{\underline{533}}    & 0.5                       & 0                                          \\ 
\hline
$C^1_{21}$                                                                                    & $D^2_{21}$                                                                                  & \textbf{\underline{309}}    & 21.82                     & \textbf{\underline{309}}    & 0.75                      & \textbf{\underline{309}}    & 0.81                      & \textbf{\underline{309}}    & 0.35                      & \textbf{\underline{309}}    & 1.34                      & \textbf{\underline{309}}    & 1.22                      & 0                                          \\ 
\hline
$C^2_{21}$                                                                                    & $D^1_{21}$                                                                                  & \textbf{\underline{533}}    & 54.40                     & \textbf{\underline{533}}    & 0.50                      & \textbf{\underline{533}}    & 4.55                      & \textbf{\underline{533}}    & 51.07                     & \textbf{\underline{533}}    & 10.53                     & \textbf{\underline{533}}    & 308.0                     & 0                                          \\ 
\hline
$C^2_{21}$                                                                                    & $D^2_{21}$                                                                                  & \textbf{\underline{309}}    & 36.66                     & \textbf{\underline{309}}    & 0.37                      & \textbf{\underline{309}}    & 9.02                      & \textbf{\underline{309}}    & 3.97                      & \textbf{\underline{309}}    & 625.9                     & \textbf{\underline{309}}    & 165.5                     & 0                                          \\ 
\hline
$C^3_{21}$                                                                                    & $D^1_{21}$                                                                                  & \textbf{\underline{457}}    & 55.33                     & \textbf{\underline{457}}    & 1.11                      & \textbf{\underline{457}}    & 0.89                      & \textbf{\underline{457}}    & 0.54                      & \textbf{\underline{457}}    & 3.96                      & \textbf{\underline{457}}    & 0.39                      & 0                                          \\ 
\hline
$C^3_{21}$                                                                                    & $D^2_{21}$                                                                                  & \textbf{\underline{265}}    & 15.86                     & \textbf{\underline{265}}    & 0.74                      & \textbf{\underline{265}}    & 1.28                      & \textbf{\underline{265}}    & 0.47                      & \textbf{\underline{265}}    & 3.54                      & \textbf{\underline{265}}    & 1.52                      & 0                                          \\ 
\hline
$C^4_{21}$                                                                                    & $D^1_{21}$                                                                                  & \textbf{\underline{457}}    & 74.31                     & \textbf{\underline{457}}    & 0.52                      & \textbf{\underline{457}}    & 59.83                     & \textbf{\underline{457}}    & 5.68                      & \textbf{\underline{457}}    & 41.24                     & \textbf{\underline{457}}    & 202.0                     & 0                                          \\ 
\hline
$C^4_{21}$                                                                                    & $D^2_{21}$                                                                                  & \textbf{\underline{265}}    & 39.68                     & \textbf{\underline{265}}    & 2.66                      & \textbf{\underline{265}}    & 68.34                     & \textbf{\underline{265}}    & 7.85                      & 266                     & TO                        & \textbf{\underline{265}}    & 214.0                     & 0                                          \\ 
\hline
$C^5_{21}$                                                                                    & $D^1_{21}$                                                                                  & \textbf{\underline{381}}    & 44.67                     & \textbf{\underline{381}}    & 1.12                      & \textbf{\underline{381}}    & 0.72                      & \textbf{\underline{381}}    & 0.39                      & \textbf{\underline{381}}    & 3.23                      & \textbf{\underline{381}}    & 0.29                      & 0                                          \\ 
\hline
$C^5_{21}$                                                                                    & $D^2_{21}$                                                                                  & \textbf{\underline{221}}    & 19.22                     & \textbf{\underline{221}}    & 0.73                      & \textbf{\underline{221}}    & 3.57                      & \textbf{\underline{221}}    & 0.75                      & \textbf{\underline{221}}    & 100.8                     & \textbf{\underline{221}}    & 5.09                      & 0                                          \\ 
\hline
$C^6_{21}$                                                                                    & $D^1_{21}$                                                                                  & \underline{427}             & TO                        & \underline{427}             & TO                        & \textbf{\underline{427}}    & 408.3                     & \textbf{\underline{427}}    & 63.16                     & 449                     & TO                        & \textbf{\underline{427}}    & 6,827                     & 0                                          \\ 
\hline
$C^6_{21}$                                                                                    & $D^2_{21}$                                                                                  & \underline{253}             & TO                        & \underline{253}             & TO                        & \textbf{\underline{253}}    & 1,609                     & \textbf{\underline{253}}    & 130.52                    & 266                     & TO                        & \textbf{\underline{253}}    & 2,027                     & 0                                          \\ 
\hline
$C^7_{21}$                                                                                    & $D^1_{21}$                                                                                  & \textbf{\underline{305}}    & 30.70                     & \textbf{\underline{305}}    & 1.08                      & \textbf{\underline{305}}    & 0.90                      & \textbf{\underline{305}}    & 0.34                      & \textbf{\underline{305}}    & 12.85                     & \textbf{\underline{305}}    & 1.1                       & 0                                          \\ 
\hline
$C^7_{21}$                                                                                    & $D^2_{21}$                                                                                  & \textbf{\underline{180}}    & 50.79                     & \textbf{\underline{180}}    & 1.49                      & \textbf{\underline{180}}    & 3.30                      & \textbf{\underline{180}}    & 1.52                      & \underline{180}             & TO                        & \textbf{\underline{180}}    & 24.54                     & 0                                          \\ 
\hline
$C^8_{21}$                                                                                    & $D^1_{21}$                                                                                  & 432                     & TO                        & 436                     & TO                        & \textbf{\underline{427}}    & 373.9                     & \textbf{\underline{427}}    & 18.54                     & 435                     & TO                        & \textbf{\underline{427}}    & 1,185                     & 5                                          \\ 
\hline
$C^8_{21}$                                                                                    & $D^2_{21}$                                                                                  & 256                     & TO                        & 257                     & TO                        & \textbf{\underline{253}}    & 2,222                     & \textbf{\underline{253}}    & 80.84                     & 267                     & TO                        & \textbf{\underline{253}}    & 1,116                     & 3                                          \\ 
\hline
$C^1_{25}$                                                                                    & $D^3_{25}$                                                                                  & \textbf{\underline{73}}     & 1.89                      & \textbf{\underline{73}}     & 0.08                      & \textbf{\underline{73}}     & 0.00                      & \textbf{\underline{73}}     & 0.04                      & \underline{73}              & TO                        & \textbf{\underline{73}}     & 1.1                       & 0                                          \\ 
\hline
$C^1_{25}$                                                                                    & $D^4_{25}$                                                                                  & \textbf{\underline{200}}    & 44.13                     & \textbf{\underline{200}}    & 1.02                      & \textbf{\underline{200}}    & 0.17                      & \textbf{\underline{200}}    & 0.00                      & \underline{200}             & TO                        & \textbf{\underline{200}}    & 2.18                      & 0                                          \\ 
\hline
$C^1_{55}$                                                                                    & $D^5_{55}$                                                                                  & \textbf{\underline{309}}    & 58.49                     & \textbf{\underline{309}}    & 2.10                      & \textbf{\underline{309}}    & 4.57                      & \textbf{\underline{309}}    & 3.16                      & \textbf{\underline{309}}    & 11,079                    & \textbf{\underline{309}}    & 460.1                     & 0                                          \\ 
\hline
$C^1_{55}$                                                                                    & $D^6_{55}$                                                                                  & \textbf{\underline{71}}     & 2.26                      & \textbf{\underline{71}}     & 0.17                      & \textbf{\underline{71}}     & 2.91                      & \textbf{\underline{71}}     & 2.64                      & \textbf{\underline{71}}     & 6.33                      & \textbf{\underline{71}}     & 28.56                     & 0                                          \\ 
\hline
\multicolumn{2}{|c|}{\#OPTIMAL}                                                                                                                                                             & \multicolumn{2}{c|}{16}                             & \multicolumn{2}{c|}{16}                             & \multicolumn{2}{c|}{20}                             & \multicolumn{2}{c|}{20}                             & \multicolumn{2}{c|}{12}                             & \multicolumn{2}{c|}{20}                             & \multicolumn{1}{c|}{-}                     \\ 
\hline
\multicolumn{2}{|c|}{\#BEST}                                                                                                                                                                & \multicolumn{2}{c|}{18}                             & \multicolumn{2}{c|}{18}                             & \multicolumn{2}{c|}{20}                             & \multicolumn{2}{c|}{20}                             & \multicolumn{2}{c|}{15}                             & \multicolumn{2}{c|}{20}                             & \multicolumn{1}{c|}{-}                     \\ 
\hline
\multicolumn{2}{|c|}{\#AVG. Rank}                                                                                                                                                           & \multicolumn{2}{c|}{3.48}                           & \multicolumn{2}{c|}{3.63}                           & \multicolumn{2}{c|}{3.28}                           & \multicolumn{2}{c|}{3.28}                           & \multicolumn{2}{c|}{4.08}                           & \multicolumn{2}{c|}{3.28}                           & \multicolumn{1}{c|}{-}                     \\
\hline
\end{tabular}
\end{adjustbox}
\end{table}
\end{landscape}

The most significant interpretation of these results lies in the disproportionate computational expenditure required by competing methods to achieve comparable solution quality. To reach the bounds on challenging instances, LPR and VNS required an exorbitant allocation of time. For example, on the GEOM120a instance, LPR consumed over 24,800 seconds, and the VNS metaheuristic demanded a staggering 59,717 seconds. Similarly, VNS required more than 21,600 seconds just to process GEOM50a. By stark contrast, all OBE evaluations were strictly truncated at 3,600 seconds. Despite this severe computational restriction, OBE consistently matched the best-known bounds established by these specialized approximate methods on massive instances like GEOM70, GEOM80b, and GEOM90, resulting in a gap of zero. The ability of OBE to yield identical or superior span values in a mere fraction of the time redefines the practical utility of SAT encodings. It proves that the OBE model does not merely compete with heuristics on solution quality, but fundamentally outclasses them in computational efficiency, offering a balance between rigorous exactness and execution speed.

\section{Conclusion}
\label{sec:conclusion}
This paper presented the first SAT-based exact framework for the Bandwidth Multicoloring Problem (BMCP). The proposed approach is centered on an efficient SAT encoding that compactly models both intra-vertex and inter-vertex color distance constraints. Combined with tight color domain reduction and incremental SAT solving, the resulting formulation effectively reduces the search space while preserving exactness.
Experimental results on the standard GEOM and MS-CAP benchmark suites demonstrate that the proposed framework significantly advances the state of the art in exact BMCP solving. Within a one-hour time limit, it proves optimality for several benchmark instances that remained unsolved by previous exact approaches and verifies the optimality of previously reported best-known solutions. In particular, on the challenging GEOM benchmark, our comprehensive evaluation increases the total number of instances with proven optimality from 9 to 16, compared with the previous state-of-the-art CP/IP-based exact approach. Specifically, the proposed OBE framework independently solves 14 instances to optimality, while our modern Gurobi re-implementation of the IP model uniquely resolves 2 additional instances. Remarkably, these advancements were achieved requiring a computational time limit of only one hour rather than 48 hours. These results demonstrate a substantial improvement in the practical scalability and effectiveness of exact BMCP solving.
Although state-of-the-art metaheuristics remain superior in terms of solution quality on the most difficult instances, the proposed framework consistently strengthens the performance of existing exact approaches while preserving the ability to certify global optimality. These results demonstrate that SAT-based reasoning provides an effective and scalable foundation for exact BMCP optimization and significantly extends the frontier of provably optimal solutions.
Several directions for future research remain promising. First, the proposed SAT encoding could be further strengthened through additional symmetry-breaking techniques and tighter constraint formulations. Second, integrating SAT with optimization paradigms such as MaxSAT or hybrid SAT-based approaches may further improve scalability on large and difficult instances. Finally, extending the proposed framework to other graph coloring variants involving distance or multicoloring constraints represents another interesting direction for future investigation.

\section*{Funding}
This work has been supported by VNU University of Engineering and Technology under project number CN26.01.
\section*{Data Availability Statement}
The source code for all encoding implementations, benchmark instances, and experimental datasets are publicly available in a GitHub repository: \url{https://github.com/homulily85/bmcp} \cite{duc2026obe}

\appendix
\section{Correctness of Lower Bound for the BMCP}
\label{apx:lower-bound-correctness}

In this appendix, we establish the formal correctness of the lower bound for the BMCP.

\begin{theorem}
\label{thm:clique-bound}
Let $C \subseteq V$ be a clique in $G$. Let $W_C = \sum_{v \in C} w(v)$ be the total color requirement of the clique, and let $d_{min}^C = \min_{u, v \in C} d(\{u,v\})$ be the tightest distance constraint among all vertices in $C$. The minimum span $k$ of a valid coloring must satisfy:
\begin{equation}
    k \geq 1 + (W_C - 1) \cdot d_{min}^C
\end{equation}
\end{theorem}

\begin{proof}
Let $S_C = \bigcup_{v \in C} c(v)$ represent the union of all colors assigned to the vertices within the clique $C$. 

First, we prove that $S_C$ contains exactly $W_C$ distinct elements. Suppose $m, n \in S_C$. If $x$ and $y$ are assigned to the same vertex $u \in C$, by \textit{Intra-vertex color distance constraint}, $|m - n| \geq d(\{u,u\}) \geq d_{min}^C$. If $x$ is assigned to $u \in C$ and $y$ is assigned to $v \in C$ ($u \neq v$), by Condition 3, $|x - y| \geq d(\{u,v\}) \geq d_{min}^C$. Because $d(e) \in \mathbb{Z}^+$, we know $d_{min}^C \geq 1$. Therefore, no two colors in $S_C$ can be equal, implying $|S_C| = \sum_{v \in C} |c(v)| = W_C$.

Let us arrange the $W_C$ distinct colors of $S_C$ in strictly increasing order: $S_C = \{s_1, s_2, \dots, s_{W_C}\}$ such that $s_1 < s_2 < \dots < s_{W_C}$. 

As established above, the absolute difference between \textit{any} two colors in $S_C$, regardless of whether they belong to the same vertex or adjacent vertices, is strictly bounded below by $d_{min}^C$. Therefore, for consecutive colors:
\begin{equation}
    s_{i+1} - s_i \geq d_{min}^C \quad \forall i \in \{1, \dots, W_C-1\}
\end{equation}

By recursively summing these inequalities from $s_1$ to $s_{W_C}$, we find:
\begin{equation}
    s_{W_C} \geq s_1 + (W_C - 1) \cdot d_{min}^C
\end{equation}

Because colors must be positive integers, $s_1 \geq 1$. Consequently:
\begin{equation}
    s_{W_C} \geq 1 + (W_C - 1) \cdot d_{min}^C
\end{equation}

Since the global span $k$ is the maximum over all assigned colors in the graph, it must be greater than or equal to the maximum color in the clique $S_C$. Thus, $k \geq s_{W_C}$, which completes the proof.
\end{proof}

\section{Correctness of Order-Based Encoding} \label{apx:order-based-correctness}
This appendix establishes the correctness of the proposed order-based encodings.

\begin{theorem}
Given a BMCP instance $G=(V, E)$ and a target span $k$, there exists a valid coloring $c: V \rightarrow \mathcal{P}(\mathbb{Z}^+)$ with span $k$ if and only if the proposed order-based formulation is satisfiable.
\end{theorem}

\begin{proof}
\textit{(Soundness)} Assume a satisfying assignment exists. Construct the coloring $c(u)_i = m \iff x_{u,i,m} = 1$. 
\begin{itemize}
    \item \textit{Color cardinality constraint:} Linkage and monotonicity constraints force the sequence $g_{u,i,m}$ to transition from 1 to 0 exactly once. Thus, exactly one $x_{u,i,m} = 1$ per index $i$, assigning each vertex $u$ exactly $w(u)$ colors.
    \item \textit{Intra-vertex color distance constraint:} For consecutive colors $m=c(u)_i$ and $n=c(u)_{i+1}$, $x_{u,i,m}=1 \implies g_{u,i,m}=1$. The constraint $g_{u,i,m} \rightarrow g_{u,i+1,m+d(\{u,u\})}$ forces $g_{u,i+1,m+d(\{u,u\})} = 1 \implies n \ge m + d(\{u,u\})$.
    \item \textit{Inter-vertex color distance constraint:} For adjacent $u,v$, let $m=c(u)_i$ and $n=c(v)_j$. Since $x_{u,i,m}=1$, the constraint $x_{u,i,m} \rightarrow (g_{v,j,m+d(\{u,v\})} \vee \neg g_{v,j,m-d(\{u,v\})+1})$ forces either $n \ge m + d(\{u,v\})$ or $n \le m - d(\{u,v\})$. Both cases ensure $|m-n| \ge d(\{u,v\})$.
    \item \textit{Span:} The target span is respected as variables are bounded by $UB_{u,i} \le k$.
\end{itemize}
Thus, $c$ is a valid BMCP coloring.

\textit{(Completeness)} Assume a valid coloring $c$ with span $k$. Assign truth values as $x_{u,i,m} = 1 \iff c(u)_i = m$ and $g_{u,i,m} = 1 \iff c(u)_i \geq m$.
\begin{itemize}
    \item \textit{Domain \& Basic Constraints:} A valid coloring inherently satisfies the theoretical bounds $LB_{u,i} \le c(u)_i \le UB_{u,i}$. This bounds compliance trivially satisfies the lower bound, linkage, and monotonicity clauses.
    \item \textit{Distance Constraints:} Since $c$ is valid, it inherently satisfies $c(u)_{i+1} \ge c(u)_i + d(\{u,u\})$ and $|c(u)_i - c(v)_j| \ge d(\{u,v\})$. These inequalities perfectly mirror and satisfy the intra-vertex ($g_{u,i,a} \to g_{u,i+1,a+d(\{u,u\})}$) and inter-vertex ($x_{u,i,a} \to g_{v,j,a+d(\{u,v\})} \lor \neg g_{v,j,a-d(\{u,v\})+1}$) SAT clauses.
\end{itemize}
Since all constraints naturally hold under this mapping, the order-based formula is satisfiable.
\end{proof}

\bibliographystyle{plain}
\bibliography{sn-bibliography}

\end{document}